\documentclass[12pt,reqno,twoside,fleqn,a4paper]{amsart}%
\usepackage{amscd,amsmath,amssymb,amsthm,palatino}%
\usepackage[mathcal]{euler}

\numberwithin{equation}{section}
\makeatletter
\renewcommand{\section}{\@startsection {section}{1}{\z@}%
                                   {-3.5ex \@plus -1ex \@minus -.2ex}%
                                   {.5\linespacing}%
                                   {\normalfont\scshape\centering}}
\makeatother

\newtheorem{thm}{Theorem}[section]
\newtheorem{lem}[thm]{Lemma}
\newtheorem{cor}[thm]{Corollary}
\newtheorem{prop}[thm]{Proposition}
  
\theoremstyle{definition}
\newtheorem{definition}{Definition}[section]

\theoremstyle{remark}

\def\beq#1\eeq{\begin{equation}#1\end{equation}}
\makeatletter
\@ifpackageloaded{euler}{
 
 \newcommand{\onto}{\to\mkern-14mu\to}
 \def\rightarrowfill@#1{\m@th\setboxz@h{$#1\relbar$}\ht\z@\z@
   $#1\copy\z@\mkern-6mu\cleaders
   \hbox{$#1\mkern-2mu\box\z@\mkern-2mu$}\hfill
   \mkern-6mu\mathord\rightarrow$}
 \def\leftarrowfill@#1{\m@th\setboxz@h{$#1\relbar$}\ht\z@\z@
   $#1\mathord\leftarrow\mkern-6mu\cleaders
   \hbox{$#1\mkern-2mu\copy\z@\mkern-2mu$}\hfill
   \mkern-6mu\box\z@$}
 \def\B@R#1#2{\raisebox{-.07ex}{$#1#2$}\mkern-6mu}
 \renewcommand{\hbar}{{\mspace{1mu}\mathpalette\B@R{\mathchar'26}h}}

}{
 
 \DeclareMathSymbol{\onto}{\mathrel}{AMSa}{"10}
 \renewcommand{\hbar}{{\mathchar'26\mkern-9muh}}

}
\makeatother
\newcommand{\C}{\mathcal{C}}
\newcommand{\co}{\mathbb{C}}
\newcommand{\cs}{\mbox{\upshape C}\ensuremath{{}^*}}

\newcommand{\R}{\mathbb{R}}
\newcommand{\Z}{\mathbb{Z}} 
\newcommand{\into}{\hookrightarrow}

\DeclareMathOperator{\rk}{rk}
\newcommand{\inner}{\mathbin{\raise1.5pt\hbox{$\lrcorner$}}}
\newcommand{\binner}{\mathbin{\raise1.5pt\hbox{$\llcorner$}}}
\DeclareMathOperator{\tr}{tr}

\newcommand{\abs}[1]{\lvert#1\rvert}
\newcommand{\Kahler}{K\"ahler} 
\newcommand{\A}{\mathcal{A}}
\renewcommand{\AA}{\mathbb{A}}
\DeclareMathOperator{\Image}{Im}
\renewcommand{\Im}{\Image}
\newcommand{\Qh}{Q_\hbar}
\newcommand{\Or}{\mathcal O}
\newcommand{\Norm}[1]{\left\|#1\right\|}
\newcommand{\norm}[1]{\lVert#1\rVert}
\DeclareMathOperator{\Mat}{Mat}
\newcommand{\N}{\mathbb N}
\newcommand{\Abs}[1]{\left|#1\right|}
\DeclareMathOperator{\Spec}{Spec}
\newcommand{\Hi}{\mathcal H}
\newcommand{\Li}{\mathcal L}
\newcommand{\starn}{\mathbin{*^{\scriptscriptstyle{n}}}}
\newcommand{\starpn}{\mathbin{*^{\prime\scriptscriptstyle{n}}}}
\newcommand{\su}{\mathfrak{su}}
\DeclareMathOperator{\SU}{SU}

\DeclareMathOperator{\Tr}{Tr}
\newcommand{\nTr}{\mathop{\widetilde\Tr}}
\newcommand{\Gn}{\Gamma^{(n)}_{\scriptscriptstyle Q}}
\newcommand{\Gnp}{\Gamma^{(n)}_{\scriptscriptstyle Q'}}

\newcommand{\eb}{e_{\mathrm{Bott}}}
\newcommand{\Ahat}{\hat{A}}
\DeclareMathOperator{\ch}{ch}

\DeclareMathOperator{\dist}{dist}

\newcommand{\T}{\mathbb T}
\newcommand{\deRham}{de\,Rham}
\newcommand{\thn}{\vartheta_{\scriptscriptstyle{(n-2)}}}

\title{An Obstruction to Quantization of the Sphere}
\author{Eli Hawkins}

\begin{document}
\maketitle
\begin{center}
\vspace{-4ex}
\emph{\small Institute for Mathematics, Astrophysics, and Particle Physics}\\
\emph{\small Radboud University Nijmegen, The Netherlands}\\
{\small mrmuon@mac.com}\\
\end{center}

\begin{abstract}
In the standard example of strict deformation quantization of the symplectic sphere $S^2$, the set of allowed values of the quantization parameter $\hbar$ is not connected; indeed, it is almost discrete. Li recently constructed a class of examples (including $S^2$) in which $\hbar$ can take any value in an interval, but these examples are badly behaved. 
Here, I identify a natural additional axiom for strict deformation quantization and prove that it implies that the parameter set for quantizing $S^2$ is never connected.
\end{abstract}

\section{Introduction}
The standard geometric quantization construction for a manifold $M$ with symplectic form $\omega$ uses a line bundle over $M$ with curvature equal to $\omega/\hbar$. Such a line bundle exists if and only if the \deRham\ cohomology class $[\frac{\omega}{2\pi \hbar}]\in H^2(M)$ is integral. Unless the symplectic form is exact, this greatly restricts the allowed values of $\hbar$. At best, there is some maximum $\hbar_0$, and the allowed values of $\hbar$ are $\hbar_0,\hbar_0/2,\hbar_0/3,\dots,0$. At worst, there does not exist any $\hbar$ satisfying this integrality condition; for example, this is the case with the Cartesian product of two symplectic spheres whose symplectic volumes differ by an irrational factor.

Of course, this is just an obstruction to a particular construction. It says nothing about whether $\hbar$ is so restricted for \emph{any} quantization of $M$. A more general existential result was found by Fedosov \cite{fed} in his study of ``asymptotic operator representations'' (AOR's) of formal deformation quantizations. Let $\theta= [\frac\omega{2\pi \hbar}] + \dots \in \hbar^{-1}H^2(M)[[\hbar]]$ be the characteristic class of a formal deformation quantization. Under some assumptions about the trace, he showed that as $\hbar\to0$, $\theta$ must asymptotically satisfy integrality conditions at the allowed values of $\hbar$.
This is an interesting result, but there are many examples of quantization in which Fedosov's assumptions are not true and $\hbar$ violates his integrality conditions. 

One way to construct such quantizations is to consider the universal covering space $\tilde M$. If we can construct a quantization of $\tilde M$ that is covariant under the action of $\pi_1(M)$, then this should descend to a quantization of $M$. This implies that requiring $[\frac\omega{2\pi \hbar}]$ (or $\theta$) to be integral on $M$ is too restrictive. Instead, we should only expect it to give an integral cohomology class on $\tilde M$.

This is a weaker condition, because $H^2(\tilde M)$ is generally smaller than $H^2(M)$. To be precise, for a path-connected manifold, $H_2(\tilde M) = \pi_2(M)$; so, $\frac\omega{2\pi \hbar}$ gives an integral cohomology class on $\tilde M$ if it pairs integrally with $\pi_2(M)$, or equivalently, for any smooth map $\varphi :S^2 \to M$,
\beq
\label{integer}
\int_{S^2}\frac{\varphi^*\omega}{2\pi \hbar} \in \Z .
\eeq
In particular, there is no integrality condition if $\omega$ pairs trivially with $\pi_2(M)$.

There are many examples of quantization where $\hbar$ is unrestricted by this condition. The noncommutative torus \cite{rie4} is a quantization of the symplectic torus $\T^2$ (for which $\pi_2(\T^2)$ is trivial). Klimek and Lesniewski \cite{k-l} constructed quantizations for higher genus Riemann surfaces (where $\pi_2$ is also trivial). I constructed \cite{haw8} such a quantization for any compact \Kahler\ manifold where $\pi_2(M)$ pairs trivially with $\omega$. Natsume, Nest, and Peter \cite{n-n-p} constructed such a quantization for any symplectic manifold with $\pi_2(M)$ trivial and $\pi_1(M)$ an exact group.

On the other hand, Li \cite{li} has constructed a very large class of examples of quantizations for any Poisson (or even \emph{almost} Poisson) manifold. In these examples, $\hbar$ takes any value in an interval, so this violates all possible integrality conditions (not to mention integrability conditions). However, his examples are badly behaved. In particular, the quantum algebras are ``much too big''; the classical algebra $\C_0(M)$ is actually a subalgebra of each of the quantum algebras. Nevertheless, these examples satisfy a definition of ``strict quantization'' which had previously seemed very reasonable.

So, the problem is not to understand how integrality conditions arise from the definition of quantization. The problem is to improve the definition by identifying good properties of quantization which do imply integrality conditions.

The sphere $S^2$ obviously plays a key role in the integrality condition \eqref{integer}, so the case of the symplectic sphere is absolutely fundamental. There is no way to understand integrality conditions without understanding the case of the symplectic sphere. That is what I am considering in this paper.

In a similar vein, Rieffel \cite[Thm.~7.1]{rie4} showed that there does not exist any noncommutative $\mathrm{SO}(3)$-equivariant product on $C^\infty(S^2)$ that can be completed to a \cs-algebra. In particular, there does not exist any equivariant strict deformation quantization of $S^2$ according to his original definition. The Berezin-Toeplitz quantization of $S^2$ is equivariant and satisfies a slightly weaker definition (quantization maps are not injective). With this weaker definition, Rieffel's proof (adapted from a theorem of Wassermann \cite{was}) actually shows that the Berezin-Toeplitz quantization is essentially the unique equivariant quantization of $S^2$. That in turn implies that the set of values of $\hbar$ cannot be connected. As Rieffel writes, this theorem leaves open the question of ``deformation quantizations which need not be invariant at all''. This is what I am addressing here, although with a slightly different notion of strict deformation quantization.

As there is no one standard definition of strict deformation quantization, I begin in Section~\ref{Definitions} by reviewing some of the various definitions that appear in the literature. I then motivate and present my definition of \emph{order $n$ strict deformation quantization} and explain why anything less than second order should not really be considered quantization.

Sections \ref{Algebraic} and \ref{Geometric} are purely motivational. In Section~\ref{Algebraic}, I briefly review the algebraic index theorem which Fedosov used to obtain his integrality result. In Section~\ref{Geometric}, I summarize the standard Berezin-Toeplitz quantization construction. This shows explicitly the structure we might expect in the quantization of $S^2$.

In Section \ref{General} I give the preliminary results which are not specific to $S^2$. Some of these are minor technical lemmas, but Lemma~\ref{Integer} is the key tool and shows an integrality property of $2\times 2$ matrices over a \cs-algebra.

Finally, in Section~\ref{Obstruction}, I prove the main results. I show that for a second order strict deformation quantization of $S^2$, the set of values of $\hbar$ cannot be connected. For an infinite order strict deformation quantization of $\C^\infty(S^2)$, I show that $\hbar$ is asymptotically restricted by Fedosov's integrality condition in terms of the characteristic class of the corresponding formal deformation quantization.

\subsection{Notation}
In this paper, I will use the notations $o^n(\hbar)$ and $\Or^n(\hbar)$ a great deal. These are always in the context of maps from $\hbar\in I\smallsetminus\{0\}$ to Banach spaces. 
\begin{definition}
$a(\hbar) = o^n(\hbar)$ if 
\[
\lim_{\hbar\to0} \hbar^{-n}\Norm{a(\hbar)} =0 \mbox,
\]
and $a(\hbar)=\Or^n(\hbar)$ if $\hbar^{-n}\Norm{a(\hbar)}$ is bounded for sufficiently small $\hbar$. More frequently, I write
\[
a(\hbar) \approx b(\hbar) \mod o^n(\hbar) 
\]
if $a(\hbar)-b(\hbar) = o^n(\hbar)$. I will always use $\approx$ in this sense.
\end{definition}
Any statement involving $o^n(\hbar)$ is really a statement about a limit, but I think this notation allows things to be written more clearly. It would really be more correct to write ``$a\in o^n(\hbar)$'', but I don't want to stray \emph{too} far from existing notational conventions.

$H^2(M)$ is \deRham\ cohomology with real coefficients.

Given an algebra $A$, the algebra of $m\times m$-matrices over $A$ is denoted $\Mat_m(A)$. The partial trace $\tr:\Mat_m(A)\to A$ takes the sum of the diagonal entries.

The standard Pauli matrices $\sigma_i\in\Mat_2(\co)$ are,
\[
\sigma_1 := \begin{pmatrix}0 & 1\\1 & 0\end{pmatrix}, \quad
\sigma_2 := \begin{pmatrix}0 & -i\\i & 0\end{pmatrix}, \quad
\sigma_3 := \begin{pmatrix}1 & 0\\0 & -1\end{pmatrix} .
\]
More importantly, these are related by $\sigma_i\sigma_j = \delta_{ij}+i\epsilon^k_{\;ij}\sigma_k$; that is,
\[
\sigma_1\sigma_2=-\sigma_2\sigma_1 = i\sigma_3, \quad \sigma_1^2=1,
\]
\emph{et cetera}.

The summation convention is used for repeated indices.

\section{Definitions of Quantization}
\label{Definitions}
In order to consider properties of quantization, it is necessary to first consider the various definitions of quantization. All of the definitions for \cs-algebraic deformation quantization are variations on the ideas proposed by Rieffel \cite{rie4,rie7,rie11}. Most of these involve a continuous field of \cs-algebras, either as a given structure, or as something whose existence and uniqueness is supposed to be guaranteed by the axioms.

All of the definitions include a parameter set $I\subseteq \R$, such that $0\in I$ is an accumulation point. This is the set of values for the parameter $\hbar$, and $\hbar=0$ is the ``classical'' value.

Rieffel's definition is stated in terms of a fixed vector space $\A_0$ with an $\hbar$-de\-pen\-dent product, involution, and norm. Using these norms, the vector space can be completed to a family of \cs-algebras $A_\hbar$. Instead of starting with a fixed vector space, we can instead start with a family of normed ${}^*$-algebras $\A_\hbar$ which are identified by bijective linear ``quantization maps'' $\Qh:\A_0\to\A_\hbar$. More generally, we can consider quantizations in which these are not bijections. Either way, each ${}^*$-algebra $\A_\hbar$ is to be completed to a \cs-algebra $A_\hbar$. These are supposed to form a continuous field, uniquely defined by the requirement that for any $a\in\A_0$, $\hbar\mapsto \Qh(a)$ is a continuous section. Finally, some definitions start with the collection of \cs-algebras or continuous field as a given structure which is not necessarily defined by the quantization maps.

The third setting is the most general, so all of the definitions I am reviewing can be stated in these terms. Let $I \subseteq \R$ with $0\in I$ an accumulation point. 
Let $\{A_\hbar\}_{\hbar\in I}$ be a collection of \cs-algebras.
 Let $\A_0\subset A_0$ be a dense ${}^*$-subalgebra. Let $\Qh : \A_0\to A_\hbar$ for each $\hbar\in I$ be linear maps such that $Q_0:\A_0\into A_0$ is the inclusion. Let $\A_\hbar\subset A_\hbar$ be the subalgebra generated by $\Im\Qh = \Qh(\A_0)$.

In most cases, the collection of \cs-algebras forms a continuous field, $A$, such that  for any $a\in\A_0$, $\hbar\mapsto \Qh(a)$ defines a continuous section, denoted $Q(a)\in\Gamma(I,A)$. I also refer to $Q:\A_0\to\Gamma(I,A)$ as a \emph{quantization map}.
Compatibility with a continuous field structure implies that for any $a,b,c\in\A_0$,
\[
\Norm{\Qh(a) + \Qh(b)^*} \text{ and }
\Norm{\Qh(a) \Qh(b) - \Qh(c)}
\]
are continuous functions of $\hbar\in I$. Conversely, these conditions imply \cite{rie11} that the \cs-algebras generated by each $\Im \Qh$ form a unique continuous field such that $Q(a)$ is a continuous section. This means that if we start with a given continuous field structure, then even if the quantization maps do not generate the given continuous field, they still generate a continuous subfield.

If there is not a continuous field structure, then these continuity properties are at least assumed to hold at $0\in I$.

With one exception to be discussed below, in all the definitions $\A_0$ is a commutative algebra of functions on a Poisson manifold, and the quantization maps are related to the Poisson bracket by the condition that $\forall a,b\in\A_0$
\beq
\label{Dirac}
[\Qh(a),\Qh(b)] \approx i\hbar \Qh(\{a,b\}) \mod o^1(\hbar) .
\eeq

After that, there are several variations in the definition. The index set $I\subseteq \R$ may or may not be connected or closed. The quantization maps $\Qh$ may be injective. The image of $\Qh$ may be a subalgebra of $A_\hbar$, or it may be dense, or it may generate a dense subalgebra. The quantization maps may or may not be ${}^*$-linear.

The table on the next page summarizes the substantive variations among some of the definitions in the literature.
\begin{table}[h]
\begin{tabular}{|cc||c|c|c|c|c|c|}
\hline
&Name &  $\C$ & $I$ & 
1-1 &
Alg & dense &  $*$ \\
\hline\hline
\cite{rie11} & \begin{tabular}{c}Strict deformation\\ quantization \end{tabular}  & \checkmark & \begin{tabular}{c}open\\ interval\end{tabular} & \checkmark &  \checkmark & \checkmark & \\
\hline
\cite{she}&\begin{tabular}{c}Operator deformation\\ quantization \end{tabular} & \checkmark 
& $[0,\epsilon)$ & \checkmark & & $\Im\Qh$ & \checkmark \\
\hline
 \cite{lan}&Strict quantization & & & & & $\Im \Qh$ & \checkmark  \\
\hline
 \cite{lan}& \begin{tabular}{c}Strict deformation\\ quantization \end{tabular} & & & \checkmark & \checkmark & \checkmark & \checkmark \\
\hline
\cite{lan}& Continuous quantization & \checkmark & & & & & \checkmark \\
\hline
\cite{l-r} & \begin{tabular}{c}Strict deformation\\ quantization \end{tabular} &\checkmark& & & & & \checkmark \\
\hline
\cite{l-r} & \begin{tabular}{c}Semi-strict\\ deformation quantization\end{tabular} & &  & & \checkmark & \checkmark & \\
\hline
\cite{n-n-p} & Strict quantization &\checkmark & $[0,\epsilon)$ & & & $\A_\hbar$ & \\
\hline
\cite{haw8} & \begin{tabular}{c}Strict deformation\\ quantization \end{tabular} & \checkmark & closed & & & $\A_\hbar$ & \\
\hline
\cite{li}& Strict quantization  &\checkmark & closed & & & & \\
&Faithful & & & \checkmark & & & \\
&Hermitian & & & & & & \checkmark \\
\hline
\cite{li} & \begin{tabular}{c}Strict deformation\\ quantization \end{tabular} & \checkmark & closed & \checkmark & \checkmark & & \\
\hline
\end{tabular}\medskip

\begin{description}
\item[$\C$] The \cs-algebras may form a continuous field. Otherwise, continuity is only required at $\hbar=0$.
\item[$I$] Any additional assumptions on the index set $I\subseteq \R$.
\item[1-1] The quantization maps may be assumed to be injective $\Qh:\A_0\into A_\hbar$.
\item[Alg] The image $\Im\Qh\subseteq A_\hbar$ may be a ${}^*$-subalgebra.
\item[dense] If the image is a ${}^*$-subalgebra, then it may be dense in $A_\hbar$. Otherwise, $\Im \Qh$ may be dense, or the ${}^*$-subalgebra $\A_\hbar$ generated by it may be dense.
\item[$*$] The quantization maps $\Qh$ may be ${}^*$-linear, i.e., $\Qh(a^*)=\Qh(a)^*$.
\end{description}
\end{table}

As this table should make clear, there is significant variation in both the definitions and terminology, and these are not correlated. There is no consistent feature of ``strict deformation quantization'' as opposed to ``strict quantization''. In the absence of a definitive definition, I shall consider the merits of each of these possible axioms.

First, I favor requiring a continuous field structure, because the idea of strict deformation quantization is to \emph{continuously} deform from $A_0$ to $A_\hbar$. I can see no aesthetic justification for requiring continuity at $0$ without requiring it at all $\hbar\in I$.

Some authors require the index set $I$ to be an interval. This is a natural requirement, given the idea of deformation. However, in practice it is not a good assumption for quantization. The Berezin-Toeplitz quantization of a compact \Kahler\ manifold \cite{b-m-s} is one of the best behaved examples of quantization, and in that case
\[
I = \left\{0,\dots,\tfrac13,\tfrac12,1\right\}
\]
which is very different from an interval. On the other hand, it is technically awkward to work with a continuous field over a pathological space, so I will assume that $I\subseteq\R$ is locally compact. 

The example of Berezin-Toeplitz quantization violates another of these tentative axioms: The quantization maps are not injective. This is unavoidable because the algebras $A_{\hbar>0}$ are finite dimensional. Note that the continuous field structure implies a sort of asymptotic injectivity; for any $a\neq0\in\A_0$,
\beq
\label{asymptotic.injective}
\lim_{\hbar\to0}\Norm{\Qh(a)} = \norm{a}\neq0
\eeq
so $a\not\in\ker\Qh$ for $\hbar$ sufficiently small. However, this does not imply that any of the maps $\Qh$ for $\hbar\neq 0$ are actually injective.

The most delicate issue is whether the image $\Im\Qh$ should be an algebra. This is true in the Berezin-Toeplitz example, but it is not clear if this is true in all nice examples of quantization. I also do not need this for the main result of this paper. Given this uncertainty, I am favoring the weaker definition and not assuming this property here. I will refer to a quantization with this property as \emph{algebraically closed}.

If the ${}^*$-subalgebras $\A_\hbar\subseteq A_\hbar$ generated by each $\Im\Qh$ are dense, then the continuous field structure is completely fixed by the quantization maps. This would be a nice property, but I do not need it for the main result here. I am also inclined to de-emphasize the role of any one quantization map. Instead, I think that the existence of a large class of equally good quantization maps is important. The tentative axiom that $\Im \Qh\subset A_\hbar$ is a dense subalgebra should probably be replaced by some sort of irreducibility property.

Finally, many definitions do not require the quantization map to be ${}^*$-linear. However, if the quantization map is at all well behaved with respect to the involution, then the ${}^*$-linear part $a \mapsto \tfrac12\left[Q(a)+Q(a^*)^*\right]$ will also be a perfectly good quantization map. For this reason, we can assume ${}^*$-linearity without any important loss of generality.

These axioms are actually too weak, because they do not require sufficient regularity at $\hbar=0$.

First note that compatibility with the continuous field structure implies that $\forall a,b\in\A_0$,
\beq
\label{order0}
\Qh(a) \Qh(b) \approx \Qh(ab) \mod o^0(\hbar) .
\eeq

I mentioned above that one definition does not include eq.~\eqref{Dirac}. That is Rieffel's definition \cite{rie11} for strict deformation quantization of a (possibly) noncommutative algebra. In that case, one cannot assume that the commutator vanishes at $\hbar=0$. Instead, Rieffel requires that
\beq
\label{order1}
\Qh(a) \Qh(b) \approx \Qh[ab + i \hbar C_1(a,b)] \mod o^1(\hbar) 
\mbox,
\eeq
where $C_1$ is a given Hochschild $2$-cocycle. The cohomology class of $C_1$ in $H^2(\A_0,\A_0)$ is required to be a ``noncommutative Poisson structure'' \cite{xu}; that is, to have vanishing Gerstenhaber bracket with itself. This definition is slightly awkward; it is more appropriate to require that $C_1$ belongs to a given Hochschild cohomology class, rather than being a given cocycle.

Equation \eqref{order1} requires the product to be expandable to first order in $\hbar$. Equation \eqref{Dirac} only requires this of the antisymmetric part of the product. In this way, eq.~\eqref{order1} is a slightly stronger condition. However, the case of noncommutative $\A_0$ suggests that eq.~\eqref{order1} is more natural.

Equations \eqref{order0}\ and \eqref{order1} are the beginning of a sequence of possible conditions.
\begin{definition}
A quantization map $Q$ is of \emph{order $n$} (for some $n=0,1,2,\dots$) if $\forall a,b\in\A_0$
\begin{subequations}
\label{order.n}
\beq
\Qh(a) \Qh(b) \approx  \Qh [a\starn b] \mod o^n(\hbar)
\eeq
where 
\beq
a \starn  b := ab + \sum_{j=1}^n (i\hbar)^j C_j(a,b) \mbox,
\eeq
\end{subequations}
and $C_j:\A_0\times\A_0\to\A_0$ are bilinear maps. It is of \emph{order $\infty$} if it is of order $n$ for any finite $n$.
\end{definition}
The latter definition is consistent, because order $n$ implies order $m<n$.

This definition has two aspects. If the quantization is not algebraically closed but is of order $n$, then $\Im \Qh$ is \emph{approximately} algebraically closed to order $\hbar^n$. So, this is partly a weaker version of the algebraic closedness axiom. On the other hand, if the quantization is algebraically closed and the quantization maps are injective, then the quantization actually defines a variable product on $\A_0$:
\beq
\label{Qproduct}
\Qh^{-1}[\Qh(a)\Qh(b)] ,
\eeq
and the order $n$ condition means precisely that this is an $n$ times differentiable function of $\hbar$ at $0$; in fact, $a \starn b$ is the $n$'th order Taylor expansion of \eqref{Qproduct} about $0$. So, ``order $n$'' is both an algebraic closedness and a differentiability condition. 

The same aesthetic principle that suggests using a continuous field of \cs-algebras implies that if we impose a differentiability condition at $0$, we should impose such a condition everywhere. However, this condition will not be needed here, so I will leave that discussion to a future paper.

\begin{definition}
Given a continuous field and a quantization map, $Q:\A_0\to \Gamma(I,A)$, define $\Gn(I,A)$ to be the set of sections $a \in \Gamma(I,A)$ for which there exist elements $a_0,\dots,a_n\in\A_0$ such that
\beq
\label{expansion}
a(\hbar) \approx \Qh(a_0 + a_1 \hbar + \dots + a_n\hbar^n) \mod o^n(\hbar) .
\eeq
\end{definition}
Note that a quantization map $Q$ is of order $n$ if and only if $\Gn(I,A)\subset \Gamma(I,A)$ is a subalgebra.
The algebra $\Gn(I,A)$ should be thought of as the set of sections of $A$ that are $n$-times differentiable at $\hbar=0$. This lemma gives the analogue of a Taylor expansion with the quantization map playing the role of a flat connection.
\begin{lem}
\label{Formal}
Given a quantization map $Q$ and a section $a\in \Gn(I,A)$, there exists a \emph{unique} degree $n$ polynomial satisfying eq.~\eqref{expansion}
\end{lem}
\begin{proof}
Existence is guaranteed by the definition of $\Gn(I,A)$.

If uniqueness is violated then there exist $a_0,\dots,a_n\in\A_0$ (not all $0$) such that 
\[
\Qh(a_0+a_1\hbar+\dots+a_n\hbar^n) = o^n(\hbar) .
\]
Suppose that $a_k$ is the first nonzero term, then
\[
\Qh(a_k+ a_{k+1}\hbar + \dots + a_n\hbar^{n-k}) = o^{n-k}(\hbar) \mbox,
\]
but setting $\hbar=0$ shows that $ a_k =Q_0(a_k) =0$. By contradiction, this proves uniqueness.
\end{proof}

Li \cite{li} has constructed a disturbing class of examples which satisfy a reasonable seeming definition of ``strict quantization'' but which really should not be considered as quantizations. One of the surprising features of his construction is that it applies not just to any Poisson manifold, but to any manifold with an antisymmetric bivector field (an ``almost Poisson manifold''). The problem is that eq.~\eqref{Dirac} does not imply that the bracket $\{\;\cdot\;,\;\cdot\;\}$ satisfies the Jacobi identity.

This is a little surprising, since a commutator always satisfies the Jacobi identity. What goes wrong is that the error in eq.~\eqref{Dirac} is not controlled enough. If we apply eq.~\eqref{Dirac} to a double commutator, we only get
\[
[\Qh(a),[\Qh(b),\Qh(c)]] \approx -\hbar^2 \Qh(\{a,\{b,c\}\}) \mod o^1(\hbar)
\]
which is vacuous. The problem is that the commutator of $\Qh(a)$ with an error term of order $o^1(\hbar)$ is only guaranteed to be of order $o^1(\hbar)$, because we know nothing more about that error term.

This problem can be fixed by requiring the quantization map to be second order, as is implied by a more general result:
\begin{prop}
\label{Expand}
Given an order $n$ quantization map $Q$, the bilinear operation $\starn: \A_0\otimes\A_0 \to \A_0[\hbar]$ is uniquely determined by $Q$, and  is approximately associative: $\forall a,b,c\in\A_0$, 
\[
a\starn (b\starn c) \equiv (a\starn b)\starn c  \mod \hbar^{n+1} .
\]
\end{prop}
\begin{proof}
The $\starn$ is uniquely defined by Lemma~\ref{Formal}.

Applying eq.~\eqref{order.n} to the product of $a,b,c\in\A_0$,
\[
Q(a)[Q(b)Q(c)] \approx Q[a\starn (b\starn c)] \mod o^n(\hbar) .
\]
There is an analogous expression with the other arrangement of parenthesis. Subtracting shows that the associator satisfies,
\[
Q[a\starn (b\starn c)-(a\starn b)\starn c] = o^n(\hbar) .
\]
By Lemma~\ref{Formal}, this shows that the coefficients of the associator vanish up to order $\hbar^n$.
\end{proof}
This implies that $\starn $ defines an associative product on $\A_0[\hbar]/\hbar^{n+1}$.

\begin{cor}
If $Q$ is of order $1$, then $C_1$ is a Hochschild cocycle, defining an element of $H^2(\A_0,\A_0)$. If $\A_0$ is commutative then $\{a,b\}:=C_1(a,b)-C_1(b,a)$ is a derivation in both arguments.
\end{cor}
The point is that this alone does not imply the Jacobi identity. The commutator does satisfy the Jacobi identity. The Jacobi identity for the bracket should appear at order $\hbar^2$, but if $Q$ is only of order $1$, then the error terms are of order $o^1(\hbar)$.

\begin{cor}
If $\Qh$ is of order $2$, then the cohomology class of $C_1$ has vanishing Gerstenhaber bracket $[C_1,C_1] = 0 \in H^3(\A_0,\A_0)$ (i.e., it is a noncommutative Poisson structure \cite{xu}). In the case of $\A_0$ commutative, the bracket $\{\;\cdot\;,\;\cdot\;\}$ satisfies the Jacobi identity and is thus a Poisson bracket.
\end{cor}

It is normally assumed that a quantization of a manifold corresponds to a Poisson structure. However, if the quantization is not algebraically closed or of at least second order,  then the Jacobi identity is really an unnecessary and artificially imposed condition. For this reason, it seems that we should always require a quantization to be of order at least $2$.

To summarize, here is the main definition I will be using in this paper:
\begin{definition}
Let $\A_0$ be a ${}^*$-algebra and $\pi\in H^2(\A_0,\A_0)$ a real Hochschild cohomology class.
An \emph{order $n$ strict deformation quantization} $(I,A,Q)$ of $(\A_0,\pi)$ consists of:
a locally compact subset $I \subseteq \R$ with $0\in I$ an accumulation point,
a continuous field of \cs-algebras $A$ over $I$,
and a ${}^*$-linear map $Q:\A_0\to \Gamma(I,A)$ such that:
\begin{enumerate}
\item At $0\in I$ this is an inclusion $Q_0:\A_0\into A_0$ of $\A_0$ as a dense ${}^*$-subalgebra; 
\item $Q$ satisfies the order $n$ condition \eqref{order.n};
\item if $n\geq1$, then the order $1$ term in the expansion belongs to the given cohomology class $C_1\in \frac12\pi$.
\end{enumerate}
Such a quantization is \emph{algebraically closed} if each $\Im \Qh \subseteq A_\hbar$ is a ${}^*$-subalgebra. It is \emph{unital} if all the algebras $\A_0$ and $A_\hbar$ and the quantization maps $\Qh:\A_0\to A_\hbar$ are unital.
This is an (order $n$) \emph{strict deformation quantization of a Poisson manifold} $M$ if $\C_0(M)\subseteq A_0\subseteq \C_{\mathrm b}(M)$ and $\pi$ is the Poisson structure.
\end{definition}

As noted above, existence of a second order quantization implies the vanishing of the Gerstenhaber bracket $[\pi,\pi]=0\in H^3(\A_0,\A_0)$. For quantization of a Poisson manifold, $\pi\in H^2(\A_0,\A_0)$ implies that $\A_0$ is closed under the Poisson bracket.

This definition encompasses a spectrum of concepts. The weakest --- a $0$'th order strict deformation quantization --- is really just a continuous field with a preferred collection of sections. Anything less than second order should not really be considered quantization.

The class of examples constructed by Li \cite{li} can be seen as a demonstration of this. His construction applies to (nonintegrable) almost Poisson manifolds because the quantization map is not required to be second order. 

If a strict deformation quantization is of infinite order, then this means precisely that the product has an asymptotic expansion, which defines a formal deformation quantization. This formal product is not necessarily differential (i.~e., the $C_j$'s may not be bidifferential). However, any formal deformation quantization is equivalent (isomorphic as a $\co[[\hbar]]$-algebra) to a differential one \cite{kon}. 

In the case of Berezin-Toeplitz quantization of a compact \Kahler\ manifold, Bordemann, Meinrenken, and Schlichenmaier \cite{b-m-s} have shown that such an asymptotic expansion exists. Thus that quantization is of infinite order and in particular of second order.

\section{Algebraic Index Theorem}
\label{Algebraic}
Fedosov's integrality results are stated in the context of formal deformation quantization and follow from the algebraic index theorem.

A formal deformation quantization of a Poisson manifold \cite{fed,kon,wei} is an associative product on the space $\C^\infty(M)[[\hbar]]$ of formal power series that reduces to ordinary multiplication modulo $\hbar$ and whose commutator is given to first order by the Poisson bracket. Write $\AA^\hbar$ for the space $\C^\infty(M)[[\hbar]]$ with this product.

It is also possible to consider just an order $n$ formal deformation quantization --- that is, an associative product on $\C^\infty(M)[\hbar]/\hbar^{n+1}$. Proposition~\ref{Expand} shows that an order $n$ strict deformation quantization determines an order $n$ formal deformation quantization.

Two formal deformation quantizations are \emph{equivalent} if the algebras are isomorphic as $\co[[\hbar]]$-algebras. The isomorphism is a $\co[[\hbar]]$-linear map $G: \C^\infty(M)[[\hbar]]\to \C^\infty(M)[[\hbar]]$ intertwining the products. 

A formal deformation quantization of a symplectic manifold $(M,\omega)$ determines a characteristic class $\theta \in \hbar^{-1}H^2(M)[[\hbar]]$ which always equals $\frac{[\omega]}{2\pi\hbar}$ plus terms of nonnegative degree. Formal deformation quantizations of $(M,\omega)$ are classified by $\theta$ up to equivalences with $G$ the identity modulo $\hbar$ \cite{wei,n-t3}.

A formal deformation quantization of $(M,\omega)$ has a unique $\co[[\hbar]]$-linear trace (modulo normalization) \cite{fed}
\[
\Tr_\hbar : \AA^\hbar \to \hbar^{-\frac12\dim M} \co[[\hbar]] .
\]
There is a natural choice of normalization.

This trace is the setting for the algebraic index theorem. The statement of the theorem is slightly simpler if $M$ is compact. Suppose that $e_0=e_0^2\in \Mat_m[\C^\infty(M)]$ is an idempotent matrix of smooth functions. Then there exists an idempotent $e_\hbar\in\Mat_m(\AA^\hbar)$ such that $e_\hbar \equiv e_0 \mod \hbar$. For any such idempotent, the algebraic index theorem \cite{fed,n-t3} shows that its trace is 
\beq
\label{index}
\Tr_\hbar e_\hbar = \int_M e^{\theta} \wedge \Ahat(TM) \wedge \ch(e_0) .
\eeq
Here, $\ch e_0\in H^\bullet(M)$ is the Chern character of the image vector bundle $e_0\,(\co^m\times M)$. Note that this trace does not depend upon the choice of $e_\hbar$. It only depends upon the topology of $M$, the characteristic class $\theta$, and the class of $e_0$ in $K^0(M)$.

Now, let $M=S^2\subset \R^3$ be the unit sphere in Euclidean space and take the induced volume form as symplectic form. Because this is $2$-dimensional, the $\Ahat$ class is trivial and eq.~\eqref{index} simplifies to
\begin{align*}
\Tr_\hbar e_\hbar &= \int_{S^2} e^{\theta} \wedge \ch(e_0)
= \int_{S^2} (1+\theta)\wedge (\rk e_0 + c_1(e_0)) \\
&= \rk e \int_{S^2} \theta + \int_{S^2} c_1(e_0) .
\end{align*}
In particular, the normalization of the trace is
\beq
\label{Tr1}
\Tr_\hbar 1 = \int_{S^2} \theta .
\eeq

The $K$-theory $K^0(S^2)$ is generated by the classes of $1$ and the Bott projection. Let $x^i$ be the Cartesian coordinates on $S^2\subset\R^3$. Let $\sigma_i\in \Mat_2(\co)$ be the Pauli matrices. The Bott projection $\eb\in\Mat_2[\C^\infty(S^2)]$ is
\beq
\label{Bott}
\eb := \tfrac12(1 + \sigma_i x^i) 
= \frac12 +\frac12\begin{pmatrix}
x^3 & x^1 - i x^2\\
x^1 + i x^2 & - x^3
\end{pmatrix}  .
\eeq
This has rank $1$ and Chern number $\int_{S^2}c_1(\eb) = 1$, so
\beq
\label{TrBott}
\Tr_\hbar e_\hbar = 1+\int_{S^2} \theta  .
\eeq
Because $\rk \eb=1$, this is actually not the best choice of generator for $K$-theory. It would be better to use $[\eb]-[1] \in K^0(S^2)$, which actually lies in the reduced $K$-theory. This corresponds to $\Tr_\hbar e_\hbar - \Tr_\hbar 1 = 1$. The class $[\eb]-[1]$ is the protagonist of this paper.

Fedosov \cite{fed} has applied this algebraic index theorem to find restrictions on $I$ for ``asymptotic operator representations'' (AOR's) of a formal deformation quantization. Any infinite order strict deformation quantization $(I,A,Q)$ of a symplectic manifold determines a formal deformation quantization, and if the algebras $A_\hbar$ are realized as concrete \cs-algebras $A_\hbar\subset\Li(\Hi)$, then $Q$ is an AOR of the formal deformation quantization. 

Suppose that the operators $\Qh(f)\in A_\hbar\subset\Li(\Hi)$ for $f\in\A_0$ are all trace class, and that the formal trace is the asymptotic expansion $\Tr \Qh(f)\sim \Tr_\hbar f$ of the operator trace, then for any idempotent, $\Tr_\hbar e_\hbar$ must be the asymptotic expansion of an integer valued function on $I\smallsetminus\{0\}$. This is a consequence of the fact that the operator trace of an idempotent is always an integer.

The trouble is that these assumptions do not generally hold for infinite order strict deformation quantizations of symplectic manifolds. 
In some very well behaved examples of strict deformation quantization, there does exist a unique trace, but it is not integer-valued on idempotents.
For example, in the ``noncommutative torus'' quantization of $\T^2$, for irrational values of $\hbar$, the set of traces of idempotents is dense in $[0,1]$ (see \cite{rie1}). That explains why the noncommutative torus can be defined over $I=\R$, violating Fedosov's restrictions on $\hbar$.

In order to go beyond Fedosov's result and find general restrictions on $I$, we must not assume \emph{a priori} that there exist any traces on a strict deformation quantization. Nevertheless, Fedosov's ideas do give an important heuristic guideline.

\section{Berezin-Toeplitz Quantization}
\label{Geometric}
The Berezin-Toeplitz construction \cite{b-m-s}  is a fairly general and very well behaved quantization construction. It provides some hints to the possible structure of an arbitrary strict deformation quantization of $S^2$.

Let $M$ be a compact, connected \Kahler\ manifold with symplectic form $\omega$ and such that the cohomology class $[\tfrac\omega{2\pi}]\in H^2(M)$ is integral. Then there exists a holomorphic line bundle $L\to M$ with a Hermitian inner product and curvature equal to $\omega$.

The $N$-fold tensor power $L^{\otimes N}$ is also a holomorphic line bundle with inner product and has curvature $N\omega$. Using the inner product and the symplectic volume form, we can construct a Hilbert space $L^2(M,L^{\otimes N})$ from sections of $L^{\otimes N}$. The set of holomorphic sections of $L^{\otimes N}$ is a finite dimensional Hilbert subspace $\Hi_N \subset L^2(M,L^{\otimes N})$. Let $\Pi_N : L^2(M,L^{\otimes N})\onto \Hi_N$ be the orthogonal projection onto this subspace. 

The Hilbert space $L^2(M,L^{\otimes N})$ is a module of the algebra $\C(M)$ of continuous functions. For any $f\in\C(M)$ and $\psi\in\Hi_N$, $T_N(f)\lvert\psi\rangle := \Pi_N f\lvert\psi\rangle$ defines a map $T_N:\C(M)\to \Li(\Hi_N)$ from functions to operators (matrices) on $\Hi_N$.

This defines an algebraically closed, infinite order, strict deformation quantization of $M$ as follows. The index set is
\[
I := \{0,\dots,\tfrac13,\tfrac12,1\} .
\]
The algebras are  $\A_0:=\C^\infty(M)$, $A_0:=\C(M)$, $A_{1/N} := \Li(\Hi_N)$. The quantization maps are $Q_{1/N} = T_N$ restricted to $\C^\infty(M)$. These quantization maps are actually surjective and define a unique continuous field structure on $\{A_\hbar\}_{\hbar\in I}$.

Now consider the sphere. Let $M=S^2\subset\R^3$ be the standard unit sphere in Euclidean space with the symplectic form equal to the induced volume form. This means that
\[
\int_{S^2} \omega = 4\pi
\]
so $\frac12\omega$ satisfies the integrality condition. 
With this choice, the Poisson brackets of the Cartesian coordinates (of the embedding) satisfy the standard $\su(2)$ relations.
The downside of this normalization is that we should take $L$ to have curvature $\frac12\omega$, and identify $\hbar = \frac2N$.

The sphere $S^2$ is of course an $\SU(2)$-symmetric space, and the quantization can be carried out equivariantly. 
The Hilbert space $\Hi_N$ is $N+1$-dimensional, carrying the spin-$\frac N2$ irreducible representation of $\SU(2)$.

Let $x^i\in\C^\infty(S^2)$ be the Cartesian coordinates of the embedding $S^2\subset\R^3$. The quantizations of these functions are 
\[
Q_{\hbar}(x^i) = T_N(x^i)  = \tfrac{2}{N+2} J^i
\]
where $J^i$ are the standard self-adjoint $\su(2)$-generators. 

Again, let $\sigma_i \in \Mat_2(\co)$ be the Pauli matrices and $\eb\in \Mat_2(\A_0)$ the Bott projection \eqref{Bott}. Applying the quantization map gives a matrix $\Qh(\eb) = \frac12 + \frac1{N+2} \sigma_i J^i \in \Mat_2(A_\hbar)$ with only two eigenvalues, $0$ and $\frac{N+1}{N+2}$. This is easily corrected to give a true projection,
\beq
\label{equivariant.e}
e(\hbar) = \tfrac{N+2}{2(N+1)} + \tfrac1{N+1} \sigma_i J^i  .
\eeq
From this, we can explicitly compute the trace, $\Tr e(\hbar) = N+2$. 

In fact, by comparing the traces $\Tr 1 = N+1=\frac2\hbar+1$ and $\Tr e(\hbar) = N+2=\frac2\hbar+2$ with equations \eqref{Tr1} and \eqref{TrBott}, we can deduce that if a formal deformation quantization is extracted from this Berezin-Toeplitz quantization, then its characteristic class is $\int_{S^2}\theta = \frac2\hbar+1$. See \cite{haw3} for a more general result.

The projection $e(\hbar)$ extends the Bott projection $\eb$, so the $K$-theory class $[e(\hbar)]-[1] \in K_0(A_\hbar)$ extends $[\eb]-[1]\in K^0(S^2)$. The usual trace on the matrix algebra $A_\hbar$ gives a homomorphism $[\Tr]:K_0(A_\hbar) \to \Z\subset\R$, and the trace of this particular class is
\[
[\Tr]([e(\hbar)]-[1]) = 1 .
\]
This trace (in general, any finite trace) factors through degree $0$ cyclic homology via the Chern character as
\[
[\Tr] : K_0(A_\hbar) \xrightarrow{\ch_0} HC_0(A_\hbar) \xrightarrow{\Tr} \co . 
\]
The cyclic homology group $HC_0(A_\hbar)$ is just the quotient of $A_\hbar$ by the subspace spanned by commutators. The degree $0$ part of the Chern character of an idempotent is represented by its partial trace:
\begin{definition}
The \emph{partial trace} $\tr : \Mat_m(A_\hbar)\to A_\hbar$ maps a matrix over $A_\hbar$ to the sum of its diagonal entries.
\end{definition}
The partial trace of our idempotent is $\tr e(\hbar) = 1 + \tfrac1{N+1}$, and the Chern character $\ch_0([e(\hbar)]-[1])$ is represented by
\beq
\label{constant.trace}
\tr e(\hbar) -1= \tfrac1{N+1} \in A_\hbar.
\eeq

Although \eqref{constant.trace} is a multiple of the identity, this is not typical. Because the construction of $e(\hbar)$ was $\SU(2)$-equivariant, $\tr e(\hbar)$ is inevitably invariant, and thus a multiple of the identity since $\Hi_N$ is irreducible.

Imagine now that we forget about the $\SU(2)$ action and quantize $S^2$ using a non equivariant \Kahler\ structure. We can still construct $e(\hbar)$ by the same procedure, but now $\tr e(\hbar) -1\in A_\hbar$ will (probably) not be a multiple of the identity. However, for $\hbar$ small enough, $e(\hbar)$ will have have the same class in the $K$-theory $K_0(A_\hbar)$ as the equivariant one (see \cite{haw3}). 
This implies that the (total) trace of $\tr e(\hbar) -1$ is the same. So, applying the normalized trace $\nTr : A_\hbar \to \co, a \mapsto (\Tr a)/(N+1)$ should give
\[
\nTr[\tr e(\hbar)-1] = \tfrac1{N+1}
\]
for $N$ sufficiently large.

Now what can we deduce about the value of $N$ from the \emph{existence} of the trace  on $A_\hbar$ without actually \emph{using} the trace?
The trick is to note that the normalized trace on $\A_\hbar$ is a state. This implies that if a self-adjoint element $a=a^*\in A_\hbar$ is bounded by real numbers as $\alpha\leq a \leq\beta$, then $\nTr a$ satisfies the same bound,
$\alpha \leq \nTr a \leq \beta $.
Applying this to $\tr e$, we see that if $\alpha\leq\tr e(\hbar)-1\leq\beta$ then 
\[
\alpha \leq\tfrac1{N+1} \leq\beta .
\]
So, if $\tr e$ is sufficiently close to a multiple of the identity ($\beta-\alpha\lesssim \frac12\hbar^2$), then there will be a unique value of $N\in\N$ fitting this bound. I will show in the next section (Lem.~\ref{Integer}) that an integer can be extracted from a projection in this way quite generally.

\section{General Results}
\label{General}
Before proceeding further, I will need a simple result about the spectrum of a sum of bounded operators.
\begin{lem}
\label{Spectrum}
If $a,b\in A$ are self-adjoint elements of a unital \cs-algebra, then the spectrum of the sum satisfies
\beq
\label{spectrum.bound}
\Spec(a+b) \subseteq \Spec a + [-\norm{b},\norm{b}]
\eeq
where the sum denotes the set of all sums of elements of the two sets.
\end{lem}
\begin{proof}
If $\lambda\not\in\Spec a$, then $(a-\lambda)^{-1}$ is contained in the commutative \cs-subalgebra generated by $a$ and $1$. That is the algebra $\C(\Spec a)$ of continuous functions on the spectrum. The norm of $(a-\lambda)^{-1}$ is the supremum over $\Spec a$ of the inverse of the distance to $\lambda$. Equivalently, it is the inverse of the infimum of the distance,
\[
\Norm{(a-\lambda)^{-1}} = \left(\dist[\lambda,\Spec a]\right)^{-1} .
\]

Now, suppose that $\norm b < \dist[\lambda,\Spec a]$. This implies that 
\[
\Norm{b\frac1{a-\lambda}} \leq \frac{\norm b}{\dist[\lambda,\Spec a]} <1
\]
and the power series,
\beq
\label{inverse.series}
\frac1{a-\lambda} \sum_{j=0}^\infty \left[-b\frac1{a-\lambda}\right]^j
\eeq
converges to an element of $A$.
Multiplying \eqref{inverse.series} by $a+b-\lambda$ shows that this sum must equal the inverse of $a+b-\lambda$. In particular, that inverse exists and therefore $\lambda\not\in\Spec(a+b)$.

Equivalently, $\lambda\in \Spec(a+b)$ implies $\dist[\lambda,\Spec a] \leq \norm b$. Since $a$ and $b$ are self adjoint, $\Spec a,\Spec(a+b) \subset \R$, and this gives eq.~\eqref{spectrum.bound}.
\end{proof}

The next lemma is my most important tool for proving the main results.
\begin{lem}
\label{Integer}
Let $A$ be any unital \cs-algebra, and $e=e^2=e^*\in\Mat_2(A)$ a projection. If there exist real numbers $\alpha\leq\beta\in\R$ bounding the partial trace
\[
\alpha \leq \tr(e) -1 \leq \beta
\]
then either $\alpha\leq0\leq\beta$ or there exists an integer $k\in\Z$ such that 
\[
\alpha \leq \tfrac1k \leq \beta .
\]
\end{lem}
\begin{proof}
If $\alpha\leq0\leq\beta$, then the claim is true. It is sufficient to assume  that $0<\alpha,\beta$, because we could equivalently consider the complementary projection $1-e$ instead.

If we denote the entries of this matrix as
\[
e=
\begin{pmatrix}
Z & X \\
X^* & 1-W
\end{pmatrix}
\]
then the properties $e=e^*$ and $e^2=e$ give relations among $X,Z,W\in A$,
\begin{subequations}
\label{XZW}
\begin{gather}
Z=Z^*,  \quad W=W^*, \label{self-adjoint}\\
ZX =XW, \label{commutation}\\
Z(1-Z) = X X^*, \quad W(1-W) = X^* X . \label{one}
\end{gather}
\end{subequations}

The relations \eqref{XZW} have several implications for the spectra of $Z$ and $W$.

First, equations \eqref{self-adjoint} and \eqref{one} imply that 
\[
\Spec Z, \Spec W \subset [0,1] .
\]

Second, by \eqref{one} and \eqref{commutation}, if $\lambda\not\in\Spec Z$ then
\begin{align*}
W(1-W) &= X^*X = X^*(Z-\lambda)^{-1}(Z-\lambda)X = X^*(Z-\lambda)^{-1}X(W-\lambda) .
\end{align*}
So,
\[
X^*(Z-\lambda)^{-1}X = W(1-W)(W-\lambda)^{-1}
\]
is bounded and either $\lambda=0$, $\lambda=1$, or $\lambda\not\in\Spec W$. This shows that $\Spec W \subseteq \{0,1\}\cup \Spec Z$, and by a symmetrical argument,
\beq
\label{spec1}
\Spec Z \cup \{0,1\} = \Spec W \cup \{0,1\} .
\eeq

The third relation comes from the bounds on $\tr(e) -1 = Z-W$. By Lemma~\ref{Spectrum},
\beq
\label{spec2}
\Spec Z \subset \Spec W + [\alpha,\beta] .
\eeq

Relation \eqref{spec2} implies that $\sup(\Spec Z) - \sup(\Spec W) \geq \alpha > 0$.
If $1\not\in \Spec Z$, then eq.~\eqref{spec1} implies that $\sup(\Spec W) = \sup(\Spec Z)$, therefore by contradiction $1 \in \Spec Z$. 

Combining relations \eqref{spec1} and \eqref{spec2} gives
\begin{align*}
\Spec Z &\subset (\{0,1\}\cup\Spec Z) + [\alpha,\beta] \\
&\;\; = [\alpha,\beta] \cup (\Spec Z + [\alpha,\beta])
\end{align*}
since $1+\alpha >1$. Iterating this gives
\[
\Spec Z \subset \bigcup_{k=1}^{\infty} [k\alpha,k\beta] .
\]

Now, $1\in\Spec Z$ implies that there exists $k\in\N$ such that $k\alpha\leq 1 \leq k\beta$.
\end{proof}
There is no way to avoid requiring that $A$ is unital in this theorem. Without a unit, it would be meaningless to compare an element of $A$ with real numbers.

Any idempotent matrix of functions can be extended to an idempotent over a formal deformation quantization. This lemma is the analogous result for strict deformation quantization.  
\begin{lem}
\label{Holomorphic}
If $(I,A,Q)$ is an order $n$ strict deformation quantization, and $a=a^*\in \Gn(I,A)$ such that $a(0)=a^2(0)$ is a projection, then there exists a neighborhood $I'\subseteq I$ of $0$ (in the relative topology of $I$) such that there is a projection $e=e^*=e^2 \in \Gn(I',A)$ defined by 
\beq
\label{e.integral}
e := \frac1{2\pi i}\oint \frac{d\lambda}{\lambda - a\rvert_{I'}}
\eeq
where the contour encloses the interval $(\frac12,\frac12+\frac1{\sqrt2})$.
\end{lem}
\begin{proof}
The function $\norm{a(\hbar)-a^2(\hbar)}$ is continuous and vanishes at $\hbar=0$, therefore, (by local compactness of $I$) there exists a neighborhood $I'\subseteq I$ of $0$ over which it is bounded by some $c<\frac14$.
The spectrum of $a$ restricted to $I'$ is contained in $(\frac12-\frac1{\sqrt2},\frac12)\cup(\frac12,\frac12+\frac1{\sqrt2})$, so $e\in \Gamma(I',A)$ is well defined by \eqref{e.integral}. Because $a(0)$ is already a projection, $e(0)=a(0)$.

Note that $a-a^2 \in \Gn(I,A)$, so $a-a^2=\Or^1(\hbar)$.
Consider the polynomial $p(x)= 2x^2- x^3$. Define $e_0:=a\lvert_{I'}$ and $e_{j+1} := p(e_j)$. For numbers $x\in\R$ such that $\abs{x-x^2}<\frac14$, this polynomial satisfies 
\[
\Abs{p(x)-p^2(x)} \leq \tfrac94 \Abs{x-x^2}^2 .
\]
So,
\begin{align*}
\Norm{e_j(\hbar)-e_j^2(\hbar)} 
&\leq \tfrac94 \Norm{e_{j-1}(\hbar)-e_{j-1}^2(\hbar)}^2 \\
&\leq (\tfrac32)^{2j} \Norm{e_0(\hbar)-e_0^2(\hbar)}^{2^j} 
= \Or^{2^j}(\hbar) .
\end{align*}
Similarly, 
\[
\Norm{e(\hbar)-e_j(\hbar)} \leq 2 \Norm{e_j(\hbar)-e_j^2(\hbar)} =  \Or^{2^j}(\hbar) .
\]

Because $e_j$ is just a polynomial of $e_0$, $e_j\in\Gn(I',A)$. For $2^j\geq n+1$,
\[
e \approx e_j \mod o^n(\hbar)
\]
and therefore $e\in\Gn(I',A)$.
\end{proof}

There is a fairly natural notion of equivalence of quantizations.
\begin{lem}
\label{Equivalence}
If $(I,A,Q)$ is an order $n$ strict deformation quantization and $Q':\A_0\to\Gn(I,A)$ such that $Q'_0=Q_0$, then $(I,A,Q')$ is also an order $n$ strict deformation quantization, the quantization maps $Q$ and $Q'$ determine equivalent products on $\A_0[\hbar]/\hbar^{n+1}$, and $\Gnp(I,A)=\Gn(I,A)$.
\end{lem}
\begin{proof}
Define $G(a)\in\A_0[\hbar]$ to be the degree $n$ polynomial such that
\[
Q'(a) \approx Q[G(a)] \mod o^n(\hbar) .
\]
This exists and is unique by Lemma~\ref{Formal}. Extend this map by $\co[\hbar]$-linearity to $G:\A_0[\hbar]\to\A_0[\hbar]$.

Let $\starn$ and $\starpn$ be the (approximately associative) products determined by $Q$ and $Q'$ in eq.~\eqref{order.n}. By definition $\forall a,b\in\A_0[\hbar]$,
\begin{align*}
Q[G(a\starpn b)] &=
Q'(a\starpn b) \approx Q'(a)Q'(b)  &&\bmod o^n(\hbar)\\
&\approx Q[G(a)]\, Q[G(b)] 
\approx Q[G(a)\starn G(b)] &&\bmod o^n(\hbar) .
\end{align*}
Therefore $G(a\starpn b) \approx G(a)\starn G(b) \mod \hbar^{n+1}$. Since $G$ is the identity modulo $\hbar$, it is invertible and therefore an equivalence between these products on $\A_0[\hbar]/\hbar^{n+1}$.

A section $a\in\Gn(I,A)$ has an expansion \eqref{expansion} in terms of  $Q$ (by definition). However, the inverse of $G$ makes that into an expansion in terms of $Q'$. Therefore, $\Gnp(I,A)=\Gn(I,A)$.

This shows in particular that these maps give cohomologous Hochschild cocycles $C_1$ and $C'_1$. Thus $Q'$ is a quantization for the same (noncommutative) Poisson structure on $\A_0$.
\end{proof} 

Lemma~\ref{Integer} only applies to unital \cs-algebras, so it is not directly applicable to an arbitrary strict deformation quantization of $S^2$. Fortunately, an arbitrary strict deformation quantization (of a unital algebra) can always be made unital. 
\begin{lem}
\label{Unital}
If $(I,A,Q)$ is an order $n$ strict deformation quantization of a unital algebra $\A_0\subset A_0$, then there exists an order $n$ \emph{unital} strict deformation quantization $(I',A',Q')$ of $\A_0$, where $I'\subseteq I$ is a neighborhood of $0$, $A'\subseteq A\rvert_{I'}$ is a continuous subfield of \cs-algebras, and
\[
Q' : \A_0 \to \Gn(I',A) \cap \Gamma(I',A') \mbox.
\]

\end{lem}
\begin{proof}
$1\in \A_0$ is a projection, so by Lemma~\ref{Holomorphic}, we can correct $Q(1)$ to a projection $e\in \Gn(I',A)$.
Define the algebras $A'_\hbar := eA_\hbar e$ for $\hbar\in I'$. These form a continuous subfield, because $e$ is a continuous projection. The projection $e(\hbar)$ becomes the unit of $A'_\hbar$.

Define $V := e\, Q(1) e \in \Gn(I',A')$. By construction, $V \approx 1 \mod \Or^1(\hbar)$. So,
\begin{align*}
V^{-1/2} &\approx \sum_{j=0}^n \binom{-\frac{1}2}{ j} (V-1)^j \mod \Or^{n+1}(\hbar)\\
&\;\in \Gn(I',A') .
\end{align*}

Define $Q'(a) := V^{-1/2} Q(a) V^{-1/2}$. This satisfies $Q'(a) = e\, Q'(a) = Q'(a)e$, so in fact $Q(a)\in\Gamma(I',A')$. On the unit, this gives $Q'(1) = V^{-1/2} Q(1) V^{-1/2} = e$, so $Q_\hbar:\A_0\to A'_\hbar$ is unital. 

For any $a\in\A_0$, $Q'(a)\in\Gn(I',A)$ because $V\in \Gn(I',A)$. 
\end{proof}

\section{The Obstruction}
\label{Obstruction}
Now consider some strict deformation quantization of $S^2$. As always, the dense subalgebra $\A_0\subset \C(S^2)$ must be closed under the Poisson bracket. I will also need to assume that it is closed under solution of the Poisson equation. In order to use the Bott projection, I will also require that the Cartesian coordinates of $\R^3\supset S^2$ restrict to functions in $\A_0$. 

The most natural choices of algebra $\A_0\subset \C(S^2)$ would be the algebra of smooth functions $\C^\infty(S^2)$ or the algebra of polynomials (in the coordinate functions $x^i\in \C^\infty(S^2)$). Both of these algebras satisfy all of these conditions.

In order to apply Lemma~\ref{Integer}, we need a projection whose partial trace is sufficiently close to a multiple of the identity. That is the purpose of the next lemma.  Let $\Delta$ denote the Laplacian.

\begin{lem}
\label{Squash}
If $(I,A,Q)$ is an order $n\geq2$, unital, strict deformation quantization of $S^2$,  with $x^i\in\A_0$ and 
\[
\Delta f \in \A_0 \implies f \in \A_0 \mbox,
\]
then there exists a projection $e \in \Gn[I',\Mat_2(A)]$ defined over a neighborhood $I'\subseteq I$ of $0$, such that $e(0)=\eb$ and
\[
\tr e \approx 1 + \thn^{-1} \mod o^n(\hbar) \mbox,
\]
where $\thn  \in \frac2\hbar+\co[\hbar]$ is a degree $n-2$ Laurent polynomial.
\end{lem}
\begin{proof}
The quantization maps naturally extend to matrices by linearity, giving a quantization of $\Mat_2(\A_0)$. So, Lemma~\ref{Holomorphic} applies and a projection $e$ can be constructed by applying functional calculus to $Q(\eb)$. This satisfies $e(0)=\eb$.

It is a little easier to work with the equivalent grading operator $u := 2e-1$. As any $2\times2$ matrix, it can be decomposed in terms of Pauli matrices in the form,
\[
u = \tau + \sigma_i \hat x^i 
\mbox,
\]
where $\tau,\hat x^i \in \Gn(I',A)$. 
In this case, $e(0)=\eb$ means that $\tau(0)=0$, and $\hat x^i(0)=x^i$. The identity $u^2=1$  implies the commutation relations,
\[
\tau\, \hat x^i + \hat x^i \tau 
= -\tfrac{i}2 \epsilon^{i}_{\;jk}[\hat x^j,\hat x^k] 
\approx \tfrac{\hbar}2 \epsilon^{i}_{\;jk} Q(\{x^j,x^k\}) 
= \hbar Q(x^i) \approx \hbar \hat x^i \mod \Or^2(\hbar) .
\]
Therefore $\tau \approx \frac12\hbar \mod \Or^2(\hbar)$ and 
\[
\tr e = 1 + \tau \approx 1 + \tfrac12\hbar \mod \Or^2(\hbar) .
\]
Since $\tr e\in\Gn(I',A)$, this implies that
\beq
\label{tre1}
\tr e \approx 1 + \tfrac12\hbar  + \hbar^2Q(g) \mod o^2(\hbar)
\eeq
for some $g\in \A_0$.

This shows that $\tr e$ equals a multiple of the identity to order $\Or^2(\hbar)$.  The next step is to improve $u$, so that $\tr e$ will equal a multiple of the identity to order $o^2(\hbar)$.

First, add some self-adjoint $\delta\in\Gn[I',\Mat_2(A)]$ to $u$, and then define $u'$ by ``correcting'' $u+\delta$ with functional calculus so that $u'^2=1$. My choice of $\delta$ is 
\beq
\label{delta}
\delta := Q(f) - u\, Q(f) u = u[u,Q(f)] = \tfrac12 [u,[u,Q(f)]]
\eeq
for some $f\in \A_0$. This vanishes at $\hbar=0$, so $\delta = \Or^1(\hbar)$. 

This $\delta$ was deliberately chosen such that it anticommutes 
with $u$, i.e., $u\delta=-\delta u$. So, $(u+\delta)^2 = 1 + \delta^2$ and we can write $u'$ explicitly as
\begin{align}
u' &= (1+ \delta^2)^{-1/2} (u + \delta) \nonumber\\
&\approx u + \delta - \tfrac12 u \delta^2 \mod \Or^3(\hbar) 
\label{uprime} .
\end{align}

We need to compute $\tr u'$ to second order in $\hbar$. Surprisingly, we only need the first 2 terms of \eqref{uprime}, because the partial trace of the third term is actually of order $\Or^3(\hbar)$. The reason is that $\tau$ is scalar up to first order, so $[\tau,\hat x^i]=\Or^3(\hbar)$. 
Hence, 
\begin{align*}
\tr u' - \tr u &\approx \tr (Q(f)-uQ(f)u)  
= \tfrac12 \tr [u,[u,Q(f)]]  &&\bmod o^2(\hbar) \\
&\;= [\hat x_i,[\hat x^i,Q(f)]] + [\tau,[\tau,Q(f)]] \\
&\approx -\hbar^2 Q(\{x_i,\{x^i,f\}\}) 
\approx \hbar^2 Q(\Delta f) &&\bmod o^2(\hbar) .
\end{align*}
Therefore, this will change the partial trace of the projection to 
\[
\tr e'\approx\tr e  + \frac12\hbar^2Q(\Delta f) \mod o^2(\hbar).
\]

Now we have already seen that $\tr e \approx 1+\frac12\hbar + \hbar^2 Q(g) \mod o^2(\hbar)$. We need to replace the function $g$ with its average value,
\[
c := \frac1{4\pi} \int_{S^2} g\omega .
\]
The function $g-c$ integrates to $0$ and is therefore in the image of the Laplacian. So, if we choose $f\in\A_0$ such that $-2\Delta f = g-c$, then
\[
\tr e' \approx 1+ \tfrac12\hbar + c\hbar^2 
\approx 1 + \left(\tfrac2\hbar-4c\right)^{-1} \mod o^2(\hbar).
\]

This proves the claim up to second order. To correct $u$ from order $n-1$ to order $n$, we can use exactly the same procedure, except with $u+\hbar^{n-2}\delta$.
\end{proof}

It would have been simpler to equivalently state that $\tr e$ is approximated by a polynomial $1+\tfrac12\hbar+\dots$. However, this awkward formulation of Lemma~\ref{Squash} simplifies Theorem~\ref{Disconnect} and the notation $\vartheta$ is motivated by Theorem~\ref{Infinite} below.

\begin{thm}
\label{Disconnect}
If $(I,A,Q)$ is a second order strict deformation quantization of $S^2$ such that $x^i\in \A_0$ and 
\[
\Delta f\in \A_0 \implies f \in \A_0
\]
then there is no connected neighborhood of $0$ in $I$ and in particular $I$ is not connected.
If the quantization is of order $n\geq 2$, then there exists a degree $n-2$ Laurent polynomial $\thn \in \frac2\hbar+\co[\hbar]$  such that $\forall \hbar\neq0\in I$
\beq
\label{approx.int}
\dist\left[\thn(\hbar),\Z\right] = o^{n-2}(\hbar) .
\eeq
This Laurent polynomial is unique modulo adding a constant integer. If the quantization is unital, then $\thn$ is the same as in Lemma~\ref{Squash}.
\end{thm}
\begin{proof}
The first claim will follow from \eqref{approx.int}. Note that \eqref{approx.int} is true if and only if it is true for a neighborhood $I'\subset I$ of $0$, so by Lemmas \ref{Equivalence} and \ref{Unital}, it is sufficient to prove \eqref{approx.int} for a unital quantization.

So, assume the quantization is unital and let $e$ be the projection and $\thn$ the Laurent polynomial from Lemma~\ref{Squash}. Define
\[
E(\hbar) := \Norm{\tr(e) - 1 - \thn^{-1}} .
\]
By Lemma~\ref{Squash}, this is of order $o^n(\hbar)$. For any nonzero $\hbar\in I$, Lemma~\ref{Integer} says that there exists $k\in\Z$ such that
\[
\thn(\hbar)^{-1} - E(\hbar) \leq \frac1k \leq \thn(\hbar)^{-1} + E(\hbar) .
\]
If $\hbar$ is sufficiently small, then $E(\hbar)\abs{\thn(\hbar)}<1$ and
\[
\dist\left[\thn(\hbar),\Z\right] 
\leq \Abs{\thn(\hbar) - k}
\leq \frac{E(\hbar)\thn^2(\hbar)}{1- E(\hbar)\abs{\thn(\hbar)}}
= o^{n-2}(\hbar) .
\]

Suppose that $\thn$ and $\thn'$ both satisfy eq.~\eqref{approx.int}. Then $\thn-\thn'\in\co[\hbar]$ and
\[
\dist[\thn(\hbar)-\thn'(\hbar),\Z] = o^{n-2}(\hbar) .
\]
Taking the limit $\hbar\to0$ shows that $m:=\thn(0)-\thn'(0) \in \Z$. So, 
\[
\thn(\hbar)-\thn'(\hbar)-m = \Or^1(\hbar)
\] 
is smaller than $\frac12$ for $\hbar$ sufficiently small, and thus
\[
\thn(\hbar)-\thn'(\hbar)-m = o^{n-2}(\hbar) .
\]
Because $\thn-\thn'$ is of degree $n-2$, this proves that $\thn(\hbar)=\thn'(\hbar)+m$.

For $n=2$, we can write $\thn = \frac2\hbar + c$. So, for all sufficiently large integers $k\in\Z$,
\[
2\left(k - c + \tfrac12\right)^{-1}\not\in I .
\]
Since $0\in I$ is an accumulation point, this shows that $I$ --- or any neighborhood of $0$ in $I$ --- cannot be connected.
\end{proof}

If the quantization was only first order, then we would still have that
\[
\tr e \approx 1 + \tfrac12\hbar \mod o^1(\hbar)
\]
but this would not restrict the values of $\hbar$ at all. Even if the error is of order $\Or^2(\hbar)$, this would be inadequate. In that case, $2/\hbar$ is restricted to a sequence of intervals around integers, but the width of the intervals is only bounded (rather than convergent to $0$), so the intervals may overlap.


This final theorem connects Theorem~\ref{Disconnect} with Fedosov's integrality condition.
\begin{thm}
\label{Infinite}
If $(I,A,Q)$ is an infinite order, strict deformation quantization of $\C^\infty(S^2)$ and $\theta \in \hbar^{-1} H^2(S^2)[[\hbar]]$ is the characteristic class of the corresponding formal deformation quantization, then for any $n\geq 2\in\N$ the Laurent polynomial $\thn$ in Lemma~\ref{Squash} and Theorem~\ref{Disconnect} is the truncation of
\[
\vartheta:= \int_{S^2}\theta .
\]
In other words, $\vartheta(\hbar)$ is the asymptotic expansion of some map $I\smallsetminus\{0\}\to\Z$.
\end{thm}
\begin{proof}
Let $*$ be the formal deformation quantization product obtained by asymptotically expanding $Q(f)Q(g)$.
Because $H^2(S^2)$ is $1$-dimensional, any formal deformation quantization of $S^2$ is formally equivalent to an $\SU(2)$-equivariant formal deformation quantization. Let $*'$ be such an equivalent equivariant formal product, and let $G : \C^\infty(S^2)[[\hbar]] \to \C^\infty(S^2)[[\hbar]]$ be the equivalence; that is,
\[
G(f * g) = G(f) *' G(g) .
\]

By equivariance, there exist $R=1+\dots$ and $\eta = \hbar + \dots\in\co[[\hbar]]$ such that
\[
x^i *' x^i = R^2
\]
and
\[
[x^i,x^j]_{*'} = i \eta \epsilon^{ij}_{\;\;k} x^k .
\]
However, by a minor rescaling, we can always choose $*'$ such that $R^2 = 1 - \frac14\eta^2$. 

Setting $u'_0 = \sigma_i x^i = 2\eb-1$, we can compute $u'_0 *' u'_0 = 1 + \frac14\eta^2 + \eta u'_0$. So, $(u'_0+\frac12\eta)*'(u'_0+\frac12\eta) =1$ and
\[
e'_\hbar := \eb + \tfrac14\eta
\]
is a projection for the product $*'$. This has partial trace $\tr e'_\hbar = 1 + \frac12\eta$.

Since $\tr e'_\hbar\in\co[[\hbar]]$, the algebraic index theorem shows that
\begin{align*}
\Tr_\hbar e'_\hbar &= \int_{S^2} e^{\theta} \wedge \ch \eb = 1 + \int_{S^2}\theta \\
&= (1+\tfrac12\eta) \Tr_\hbar 1 = (1+\tfrac12\eta)\int_{S^2}\theta .
\end{align*}
Therefore, $\int_{S_2}\theta = 2\eta^{-1} = 2\hbar^{-1} + \dots$.

Now let $G^{(n)}$ be the truncation of $G$ up to order $\hbar^n$. The idea is that $Q\circ G^{(n)}$ is an approximately equivariant quantization map. Define,
\[
e_0 := Q[G^{(n)}(\eb)] 
\]
and observe that
\[
e_0^2 \approx Q[G^{(n)}(\eb*'\eb)] \mod o^n(\hbar) .
\]
So if we correct $e_0$ to a projection $e$ by functional calculus, then it is approximately
\[
e \approx  e_0 + \tfrac14\eta \mod o^n(\hbar) . 
\]
This has partial trace
\begin{align*}
\tr e &= 1 + \tfrac12\eta 
\approx 1+\left(\int_{S^2}\theta_{(n-2)}\right)^{-1} \mod o^n(\hbar) ,
\end{align*}
where $\theta_{(n-2)}$ is the truncation of $\theta$ at degree $\hbar^{n-2}$.
\end{proof}

Note that this reproves Lemma~\ref{Squash}, but with different hypotheses. Lemma~\ref{Squash} applies to any order $\geq2$ quantization with minor restrictions on $\A_0$. Theorem~\ref{Infinite} requires an infinite order quantization with $\A_0=\C^\infty(S^2)$, although the proof could probably be extended to finite orders by generalizing the classification results for formal deformation quantization in the literature. I have retained Lemma~\ref{Squash} because the hypotheses are slightly more general and the proof is explicit and constructive.

\section{Remarks}
Theorem \ref{Disconnect} shows that $I\smallsetminus\{0\}$ is very close to being discrete. It is contained in a union of smaller and smaller intervals, and for an infinite order quantization, these intervals shrink faster than any power of $\hbar$.

This seems to be the strongest possible result. Consider the following variant of the Berezin-Toeplitz quantization of $S^2$. Let
\[
I := \{0\} \cup \bigcup_{N=0}^\infty [\tfrac2N-e^{-N},\tfrac2N+e^{-N}]
\]
and for $\hbar\in  [\tfrac2N-e^{-N},\tfrac2N+e^{-N}]$, $A_\hbar := \Li(\Hi_N)$, and $\Qh := T_N$.
This defines an infinite order, algebraically closed, strict deformation quantization of $S^2$. The only peculiar feature of this quantization is that $\Qh$ is locally constant away from $\hbar=0$, but an additional axiom to rule that out would probably rule out the trivial quantization of a manifold with $0$ Poisson structure. 
\medskip

The Bott projection is not the only $2\times2$ projection over $S^2$. In fact, a $2\times2$ projection on $S^2$ is equivalent to a map from $S^2$ to $\co P^1 = S^2$. The degree of the map gives the Chern number of the projection. 

Suppose that we started with a $2\times 2$ projection with Chern number $s$. Then the normalized trace would be $\nTr[e(\hbar)] = s/(N+1)$. This is not always the reciprocal of an integer. So, how does Lemma~\ref{Integer} ``know'' that we started with the Bott projection?

The answer is that $\eb$ and $1-\eb$ are the only nontrivial \emph{equivariant} $2\times2$ projections on $S^2$. Any other such projection is rather far from being equivariant. Consequently, if we quantize it, then $\tr e(\hbar)$ is never close to a multiple of the identity; its spectrum is always too wide for Lemma~\ref{Integer} to be nontrivial.
\medskip

It would be nice to show that the quantization of $S^2$ is essentially unique, i.e., that the algebras $A_{\hbar\neq 0}$ are all simple matrix algebras, and the Berezin-Toeplitz maps define continuous sections of the continuous field. However, the present axioms are not strong enough to imply this.

Given Lemma~\ref{Squash}, the proof of Lemma~\ref{Integer} can be extended slightly to show that $A_\hbar$ contains a subalgebra isomorphic to $\Mat_k(\co)$, where $k$ is the closest integer to $\Abs{\thn(\hbar)}$, and that these subalgebras form a continuous subfield of $A$ with structure equivalent to that given by Berezin-Toeplitz quantization. This is what Rieffel essentially proved in the equivariant case.

However, this does not show that $A_\hbar\cong \Mat_k(\co)$. I believe that would require some additional irreducibility axiom. The idea would be to require than any ``subquantization'' of $(I,A,Q)$ is essentially the whole thing. However, I do not know of a reasonable formulation of such an axiom. 

For example, the trivial quantization of $\R$ should not be ruled out by the axioms. In this case, $I$ is an interval, $\A_0=A_\hbar = \C_0(\R)$ and $\Qh$ is the identity map. This has proper subquantizations corresponding to bundles of open subintervals in $\R\times I$, although these are isomorphic to $(I,A,Q)$.
\medskip

It is plausible that the missing axiom is simply to assume that the quantization is algebraically closed (i.e., $\Im \Qh$ is a ${}^*$-subalgebra). Indeed, it may be that algebraic closedness would imply that $I$ is disconnected for $S^2$, without requiring the quantization to be second order. I only know that the technique I used in Lemma~\ref{Squash} requires a second order quantization. 

It would be interesting to find an alternative proof using algebraic closedness, or alternatively to find an example of a first order, algebraically closed, strict deformation quantization of $S^2$ with $I$ connected.
\medskip

The next major step should be to generalize Theorem~\ref{Infinite} to general symplectic manifolds. The key to this is perhaps some generalization of Lemma~\ref{Integer} to larger matrices.

\subsection*{Acknowledgements} 
I would like to thank Klaas Landsman and Marc Rieffel for their comments, and Ryszard Nest for encouraging me to investigate this question.

\end{document}